\documentclass[11pt, a4paper, reqno]{article}
\usepackage{amsmath,amsthm,amsfonts,amssymb}
\usepackage[english]{babel}
\usepackage{a4wide}
\usepackage{microtype}
\usepackage{tikz}
\usepackage{hyperref}
\hypersetup{colorlinks=true,citecolor=blue,linkcolor=blue,urlcolor=blue}

\DeclareMathOperator{\WRel}{{\bf \Phi}}
\DeclareMathOperator{\Feas}{{\operatorname{Feas}}}
\DeclareMathOperator{\pol}{Pol}
\DeclareMathOperator{\fpol}{fPol}
\DeclareMathOperator{\proj}{\ensuremath{e}}

\newcommand{\qnn}{\ensuremath{\mathbb{Q}_{\geq 0}}}
\newcommand{\qq}{\ensuremath{\overline{\mathbb{Q}}}}

\newcommand{\supp}{\mbox{\rm supp}}
\renewcommand{\vec}[1]{\ensuremath{\mathbf{#1}}}
\newcommand{\extend}[1]{\ensuremath{{#1}}}
\newcommand{\polk}[1]{\ensuremath{\operatorname{Pol}^{(k)}(#1)}}
\newcommand{\fpolk}[1]{\ensuremath{\operatorname{fPol}^{(k)}(#1)}}
\newcommand{\NP}{\mbox{\bf NP}}
\newcommand{\CSP}{\mathrm{CSP}}
\newcommand{\VCSP}{\mathrm{VCSP}}
\newcommand{\MCHom}{\mathrm{MinCostHom}}

\newcommand\Crestrict[2]{{
  \left.\kern-\nulldelimiterspace 
  #1 
  \right|_{#2} 
  }}

\def\GammaC{\Gamma_{c}}
\def\GammaD{\Gamma_{d}}
\def\GammaE{\Gamma_{e}}

\def\I{{\cal I}}

\def\ID{{\cal I}_{d}}
\def\IE{{\cal I}_{e}}

\def\fD{f_{d}}
\def\fE{f_{e}}

\def\DG{\mathbb{D}_\Gamma}
\def\G{\mathbb{G}}
\def\Q{\mathbb{Q}}
\def\H{\mathbb{H}}
\def\G{\mathbb{G}}
\def\zd{,\ldots,}

\newtheorem{theorem}{Theorem}
\newtheorem{lemma}{Lemma} 
\newtheorem*{lemma*}{Lemma} 
\newtheorem*{proposition*}{Proposition} 
\newtheorem*{theorem*}{Theorem} 
\newtheorem{proposition}{Proposition}
\newtheorem{corollary}{Corollary}
\theoremstyle{definition}
\newtheorem{definition}{Definition}
\newtheorem{example}{Example}

\begin{document}
\title{Binarisation for Valued Constraint Satisfaction Problems\footnote{
An extended abstract of part of this work appeared in the \emph{Proceedings of the
29th AAAI Conference on Artificial Intelligence (AAAI'15)}~\cite{ccjz15:aaai}.
Part of this work appeared in a technical report~\cite{Powell15:arxiv} and in
the fifth author's doctoral thesis~\cite{Powell15:phd}.
David Cohen, Martin Cooper, Peter Jeavons, and Stanislav \v{Z}ivn\'y were supported by EPSRC grant EP/L021226/1.
Stanislav \v{Z}ivn\'y was supported by a Royal Society University Research Fellowship. 
This project has received funding from the European Research Council (ERC) under
the European Union's Horizon 2020 research and innovation programme (grant
agreement No 714532). The paper reflects only the authors' views and not the
views of the ERC or the European Commission. The European Union is not liable
for any use that may be made of the information contained therein.
}}

\author{
David A. Cohen\\
Royal Holloway, University of London\\
\texttt{dave@cs.rhul.ac.uk}
\and
Martin C. Cooper\\
IRIT, University of Toulouse III\\
\texttt{cooper@irit.fr}
\and
Peter G. Jeavons\\
University of Oxford\\
\texttt{peter.jeavons@cs.ox.ac.uk}
\and
Andrei Krokhin\\
University of Durham\\
\texttt{andrei.krokhin@durham.ac.uk}
\and
Robert Powell\\
University of Durham\\
\texttt{robert.powell@durham.ac.uk}
\and
Stanislav \v{Z}ivn\'{y}\\
University of Oxford\\
\texttt{standa.zivny@cs.ox.ac.uk}
}

\date{}
\maketitle

\begin{abstract} 

We study methods for transforming valued constraint satisfaction problems
(VCSPs) to \emph{binary} VCSPs. First, we show that the standard \emph{dual}
encoding preserves many aspects of the algebraic properties that capture the
computational complexity of VCSPs. Second, we extend the reduction of CSPs
to binary CSPs described by Bul\'in et al.~[LMCS'15] to VCSPs. This reduction establishes that
VCSPs over a fixed valued constraint language are polynomial-time equivalent to
Minimum-Cost Homomorphism Problems over a fixed digraph.

\end{abstract}

\section{Introduction}

The valued constraint satisfaction problem (VCSP) is a general framework for
problems that involve finding an assignment of values to a set of
variables, where the assignment must satisfy certain feasibility conditions and
optimise a certain objective function. The VCSP includes as a special case the
(purely decision) constraint satisfaction problem
(CSP)~\cite{Montanari74:constraints} as well as the (purely optimisation)
minimum constraint satisfaction problem (Min-CSP), see~\cite{kz17:survey} for
a recent survey. 

Different subproblems of the VCSP can be obtained by restricting, in various
ways, the set of cost functions that can be used to express the constraints.
Such a set of cost functions is generally called a valued constraint
language~\cite{cccjz13:sicomp,kz17:survey}. For any such valued constraint
language $\Gamma$ there is a corresponding problem VCSP($\Gamma$), and it has
been shown that the computational complexity of VCSP($\Gamma$) is determined by
certain algebraic properties of the set $\Gamma$ known as  fractional
polymorphisms~\cite{cccjz13:sicomp}. The classical constraint satisfaction
problem (CSP)~\cite{Feder98:monotone} is a special case of the VCSP in which all
cost functions are relations. If a valued constraint language $\Gamma$ contains only
relations then we call $\Gamma$ a constraint language.

There has been significant progress on classifying the computational complexity
of different constraint
languages~\cite{Schaefer78:complexity,Hell90:h-coloring,Bula06:threeElement,Bulatov11:conservative}
and valued constraint
languages~\cite{kz13:jacm,hkp14:sicomp,ktz15:sicomp,tz16:jacm,Kozik15:icalp,tz17:sicomp,Kolmogorov17:sicomp}.
Most notably, it has been shown that a dichotomy for constraint languages,
conjectured by Feder and Vardi~\cite{Feder98:monotone}, implies a dichotomy for
valued constraint languages~\cite{Kolmogorov17:sicomp}. This result thus resolves
the complexity of valued constraint languages modulo the complexity of
constraint languages.

In \emph{binary} VCSPs every valued constraint involves at most two variables;
in other words, the interaction between variables is only pairwise. In this
paper we consider transformations of the general VCSP, with constraints of
arbitrary arity, to the binary VCSP. There are several motivations for studying
such reductions. 
Firstly, binary VCSPs have been extensively studied in the context of energy minimisation
problems in computer vision and machine
learning~\cite{Blake2011advances,Nowozin2014advanced} since pairwise interaction
is enough to model interesting problems. Secondly, algorithms for binary VCSPs
may be easier to design, as discussed below in the case of submodular VCSPs.
Finally, various aspects of binary VCSPs, such as the algebraic properties that
capture the complexity of valued constraint languages, may be easier to study on
binary instances. 

One important class of valued constraint languages are the \emph{submodular}
languages~\cite{schrijver03:combopt}. It is known that VCSP instances where all
constraints are submodular can be solved in polynomial time, although the
algorithms that have been proposed to achieve this in the general case
are rather intricate and difficult to
implement~\cite{iwata01:submodular,schrijver00:submodular}. 
In the special case of binary submodular constraints a
much simpler algorithm can be used to find a minimising assignment of values,
based on a standard max-flow algorithm~\cite{cohen04:maximal}. Our results in
this paper show that this simpler algorithm can be used to obtain exact
solutions to arbitrary VCSP instances with submodular constraints (from a finite
language) in polynomial time.

The more restricted question of which valued constraint languages can be transformed 
to binary valued constraint languages \emph{over the same domain} was studied in 
\cite{cjz08:tcs}. It was shown in~\cite{zcj09:dam} that there are
submodular valued constraint languages which \emph{cannot} be expressed (using
min and sum) by binary submodular languages over the same domain.

However, there are two well-known methods for transforming a non-binary CSP into a binary
one over a different domain of values; 
the \emph{dual encoding}~\cite{Dechter89:tree} and the
\emph{hidden variable encoding}~\cite{Rossi90:equivalence}. Both encode the non-binary
constraints to variables that have as domains of possible labels the valid
tuples of the constraints. That is, these techniques derive a binary encoding of a
non-binary constraint by changing the domain of the variables to an extensional
representation of the original constraints. A combination of these two
encodings, known as the double encoding, has also been
studied~\cite{Smith00:aaai}. It was observed in~\cite{Dechter00:dual} that both
of these standard encodings can be extended to valued constraints.

It is also known that any CSP with a fixed constraint language is polynomial-time
equivalent to one where the constraint language consists of a single binary
relation (i.e., a
digraph)~\cite{Feder98:monotone,Atserias08:digraph,Bulin15:lmcs}. 
Recent work by Bul\'in et al. shows that this reduction can be done in
a way that preserves certain algebraic
properties of the constraint language that are known to characterise the
complexity of the corresponding CSP~\cite{Bulin15:lmcs}. 

As our first contribution, we extend the idea of the dual encoding 
to valued constraint satisfaction 
and show that this standard encoding preserves
many aspects of the algebraic properties that capture the complexity of valued
constraint languages. In particular, we show that 
for any valued constraint language $\Gamma$ of finite size, there is a one-to-one
correspondence between the fractional polymorphisms of $\Gamma$ 
and the fractional polymorphisms of the binary language $\GammaD$ obtained by the dual encoding. 
Moreover, we show that $\GammaD$ preserves all \emph{identities} involving
(fractional) polymorphisms of $\Gamma$, where an identity is an equality between arbitrary
expressions involving only polymorphisms and all variables are universally
quantified.
A large body of research on the complexity of (valued) constraint languages
has shown that it is the identities satisfied by the (fractional) polymorphisms that
determine both the complexity and suitable algorithmic solution
techniques~\cite{kz17:survey}.

Hence, as well as providing a way to convert any
given instance of the VCSP to an equivalent binary instance, we show that the dual encoding
also provides a way to convert any valued constraint language to a binary
language with essentially the same algebraic properties, and hence essentially
the same complexity and algorithmic properties.
We remark that a similar transformation from constraint languages of arbitrary
arity to sets of unary and binary relations was used in~\cite{Barto16:sicomp} (and also
implicitly in~\cite{Barto13:cjm}), for the special case of the CSP.

While the idea of the dual encoding is very simple, the resulting $\GammaD$ contains a
single unary cost function and \emph{more than one} binary relation (in
general). However, all the binary relations that are included in $\GammaD$ are of
the same type and correspond to enforcing equality on the shared variables
between different constraints in instances of $\VCSP(\Gamma)$. 

As our second contribution, we adapt the proof from~\cite{Bulin15:lmcs} to the
VCSP framework and show that each VCSP, with a fixed valued constraint language
$\Gamma$ of finite size, is polynomial-time equivalent to a VCSP with valued
constraint language $\GammaE$, where $\GammaE$ consists of a single unary cost
function and a \emph{single} binary relation (i.e., a digraph). Problems of this
type have been studied as the Minimum-Cost Homomorphism Problem
(MinCostHom)~\cite{Gutin:mincostdichotomy,Hell12:sidma,Tak10:MinCostHom}, which
makes this result somewhat surprising as it was believed that MinCostHom was
essentially a more restricted optimisation problem than the VCSP. 

This second reduction, which we call the \emph{extended dual}, 
again preserves many aspects of the algebraic properties 
that capture the complexity of valued constraint languages.
In fact, we show that it preserves all identities involving
(fractional) polymorphisms of $\Gamma$ which are \emph{linear} and 
\emph{balanced}. These are the key properties for characterising most known
tractable cases.

However, the extended dual encoding does not preserve all identities: in particular,
it does not preserve the (unbalanced) identities defining Mal'tsev polymorphisms.
In fact it is impossible for any reduction to a single binary relation to
preserve such identities, without changing the algorithmic nature of the
problem, because it has been shown that any \emph{single binary relation} that
has a Mal'tsev polymorphism also has a majority
polymorphism~\cite{Kazda11:matsev}; the former is solved by a generalised form
of Gaussian elimination whereas the latter is solved by local consistency
operations.

In summary, our first reduction, using the dual encoding, 
transforms any valued constraint problem 
over an arbitrary valued constraint language $\Gamma$ of finite size
to a binary problem with more than one form of binary constraint, which satisfies 
all of the identities on fractional polymorphisms satisfied by $\Gamma$.
Our second reduction, using the extended dual encoding, 
transforms any valued constraint problem 
over an arbitrary valued constraint language $\Gamma$ of finite size
to a binary problem with just one form of binary constraint, which satisfies 
an important subclass of the identities satisfied by $\Gamma$.

\section{Background and Definitions}

In this section we will give the necessary background. Section~\ref{sec:vcsp}
defines the VCSP, whereas Sections~\ref{sec:algebraic} and~\ref{sec:identities}
present the basics of the algebraic approach to studying the complexity of the VCSP.

\subsection{Valued Constraint Satisfaction Problems}
\label{sec:vcsp}

Throughout the paper, let $D$ be a fixed finite set and let
$\qq=\mathbb{Q}\cup\{\infty\}$ denote the set of rational numbers with
(positive) infinity.
For any $m$-tuple $\vec{x}\in D^m$ we will write $\vec{x}[i]$ for its $i$th component.

\begin{definition}
An $m$-ary \emph{cost function} over $D$ is any mapping $\phi:D^m\to\qq$.
We denote by $\WRel_D^{(m)}$ the set of all $m$-ary cost functions and let
$\WRel_D=\bigcup_{m\ge 1}{\WRel_D^{(m)}}$.
\end{definition}

We call $D$ the \emph{domain}, the elements of $D$ \emph{labels} (for variables),
and we say that the cost functions in $\WRel_D$ take \emph{values}
(which are elements of $\qq$).

We denote by $\Feas(\phi)=\{\vec{x}\in D^m\:|\:\phi(\vec{x})<\infty\}$ the
underlying \emph{feasibility relation} of a given $m$-ary cost function.
A cost function $\phi:D^m\to\qq$ is called \emph{finite-valued} if $\Feas(\phi)=D^m$.

It is convenient to highlight the special case when the values taken by a cost
function are restricted to 0 and $\infty$.
\begin{definition}
Any mapping $\phi:D^m\to\{0,\infty\}$ will be called a \emph{crisp cost
function}
(or simply a \emph{relation}) and will be identified with the set 
$\{\vec{x}\in D^m\mid\phi(\vec{x})=0\}$.
\end{definition}

\begin{definition}
Let $X=\{x_1,\ldots, x_n\}$ be a set of variables. A \emph{valued constraint} over $X$ is an expression
of the form $\phi(\vec{x})$ where $\phi\in \WRel_D^{(m)}$ and $\vec{x}\in X^m$,
for some positive integer $m$.
The integer $m$ is called the \emph{arity} of the constraint,
the tuple $\vec{x}$ is called the \emph{scope} of the constraint,
and the cost function $\phi$ is called the \emph{constraint cost function}.
\end{definition}

\begin{definition}
An instance $\I$ of the \emph{valued constraint satisfaction problem} (VCSP) is specified
by a finite set $X=\{x_1,\ldots,x_n\}$ of variables, a finite set $D$ of labels,
and an \emph{objective function} $\Phi_\I$
expressed as follows:
\begin{equation}
\Phi_\I(x_1,\ldots, x_n)=\sum_{i=1}^q{\phi_i(\vec{x}_i)}
\label{eq:sepfun}
\end{equation}
where each $\phi_i(\vec{x}_i)$, $1\le i\le q$, is a valued constraint over $X$.
Each constraint can appear multiple times in $\Phi_\I$.

Any assignment of labels from $D$ to the variables of $X$ for which $\Phi_I$ 
is finite will be called a \emph{feasible solution} to $\I$.
The goal is to find a feasible solution that minimises $\Phi_\I$.
\end{definition}

\begin{definition}
Any set $\Gamma\subseteq\WRel_D$ of cost functions 
on some fixed domain $D$ is called a \emph{valued constraint language}, or
simply a \emph{language}.

We will denote by $\VCSP(\Gamma)$ the class of all VCSP instances in which the
constraint cost functions are all contained in $\Gamma$.
\end{definition}

The classical constraint satisfaction problem
(CSP) can be seen as a special case of the VCSP in which all
cost functions are crisp (i.e., relations).
A language containing only crisp cost functions is called crisp.

A language $\Gamma$ is called \emph{binary} if all cost functions from $\Gamma$ are of arity at most two.

\subsection{Fractional Polymorphisms}
\label{sec:algebraic}

Over the past few years there has been considerable progress in investigating
the complexity of different kinds of constraint satisfaction problems and valued
constraint satisfaction problems by looking at the algebraic properties of the
relations and cost functions that define the constraints and valued
constraints~\cite{Jeavons97:closure,Feder98:monotone,Bulatov05:classifying,cccjz13:sicomp}
resulting in strong complexity
classifications~\cite{kz13:jacm,hkp14:sicomp,tz16:jacm,Kolmogorov17:sicomp}.
We present here some of the tools used in this line of work.

We first need some standard terminology.
A function $f : D^k \rightarrow D$ is called a $k$-ary \emph{operation} on $D$.
For any tuples $\vec{x}_1,\ldots,\vec{x}_k\in D^m$, we denote by
$\extend{f}(\vec{x}_1,\ldots,\vec{x}_k)$ the tuple in $D^m$ obtained by applying $f$ to
$\vec{x}_1,\ldots,\vec{x}_k$ componentwise.
\begin{definition}
\label{def:polymorphism}
Let $\phi: D^m \rightarrow \qq$ be a cost function.
An operation $f:D^k\to D$ is a \emph{polymorphism} of $\phi$ if,
for any $\vec{x}_1,\ldots,\vec{x}_k \in \Feas(\phi)$
we have $\extend{f}(\vec{x}_1,\ldots,\vec{x}_k)\in \Feas(\phi)$.

We denote by $\pol(\Gamma)$ the set of all operations on $D$ which are polymorphisms of all
$\phi \in \Gamma$.
We denote by $\polk{\Gamma}$ the $k$-ary operations in $\pol(\Gamma)$.
\end{definition}
\noindent
The $k$-ary \emph{projections}, defined for all $1\leq i\leq k$, are the
operations $\proj^{(k)}_i$ such that $\proj^{(k)}_{i}(x_1,\ldots,x_k)=x_i$.
It follows directly from Definition~\ref{def:polymorphism} that all projections are polymorphisms of all valued constraint languages.

Polymorphisms are sufficient to analyse the complexity of the CSP,
but for the VCSP, it has been shown that in general
we need a more flexible notion that assigns weights to a collection of
polymorphisms~\cite{cccjz13:sicomp,fz16:toct}.

\begin{definition}
\label{def:fpol}
Let $\phi:D^m\rightarrow\qq$ be a cost function. A probability distribution
$\omega$ on the set of $k$-ary polymorphisms of $\phi$ (i.e.,
$\omega:\polk{\phi}\to\qnn$ with $\sum_{f\in\polk{\phi}}\omega(f)=1$) 
is called a $k$-ary \emph{fractional polymorphism} of $\phi$ if 
for any $\vec{x}_1,\ldots,\vec{x}_k \in \Feas(\phi)$
\begin{equation}
\label{eq:fpol}
\sum_{f\in \polk{\phi}} \omega(f)\phi(\extend{f}(\vec{x}_1,\ldots,\vec{x}_k))
\ \leq\ \frac{1}{k}\sum_{i=1}^k\phi(\vec{x}_i)\,.
\end{equation}
We denote by $\fpolk{\Gamma}$ the set of $k$-ary fractional polymorphisms of
all $\phi\in\Gamma$ and set $\fpol(\Gamma)=\bigcup_{k\geq 1}\fpolk{\Gamma}$.
\end{definition}
For any $\omega\in\fpol(\Gamma)$ we denote by $\supp(\omega)$ 
the set $\{f\in\polk{\phi} \mid \omega(f) > 0\}$ and 
define $\supp(\Gamma)=\bigcup_{\omega\in\fpol(\Gamma)}\supp(\omega)$.

\begin{example}
\label{ex:submodular}
Let $D=\{0,1\}$. Let $\Gamma$ be the set of cost functions 
$\phi:D^m\to\qq$ that admit $\omega_{sub}$ as a
fractional polymorphism, where $\omega_{sub}$ is defined by
$\supp(\omega_{sub}) = \{\min,\max\}$ and 
$\omega_{sub}(\min)=\omega_{sub}(\max)=\frac{1}{2}$;
here $\min$ and $\max$
are the binary operations returning the smaller and larger of their two arguments,
respectively, with respect to the usual order $0<1$.

In this case $\Gamma$ is precisely the well-studied class of
submodular set functions~\cite{schrijver03:combopt}.
\end{example}

\subsection{Identities and Rigid Cores}
\label{sec:identities}

Many important properties of polymorphisms can be specified by
\emph{identities}, i.e., equalities of terms that hold for all choices of the
variables involved in them. 
More formally, an operational signature
is a set of operation symbols with arities assigned to them and an identity is an
expression $t_1=t_2$ where $t_1$ and $t_2$ are terms in this signature. 

Here are some examples of important properties of operations that are specified by identities:
\begin{itemize}
\item An operation $f$ is \emph{idempotent} if it satisfies the identity $f(x,\ldots,x)=x$.
\item A $k$-ary ($k\ge 2$) operation $f$ is \emph{weak near unanimity (WNU)} if it is idempotent and satisfies the identities
$f(y,x,\dots,x,x)=f(x,y,\dots,x,x)=\cdots=f(x,x,\dots,x,y).$
\item A $k$-ary ($k\ge 2$) operation $f$ is \emph{cyclic} if $f(x_1,x_2,\dots,x_k)=f(x_2,\dots,x_k,x_1)$.
\item A $k$-ary ($k\ge 2$) operation $f$ is \emph{symmetric} if $f(x_1,\dots,x_k)=f(x_{\pi(1)},\dots,x_{\pi(k)})$ for each permutation $\pi$ on $\{1,\dots,k\}$.
\item A $k$-ary ($k\ge 3$) operation $f$ is \emph{edge} if
\[
f(y,y,x,x,\ldots,x)=f(y,x,y,x,x,\ldots,x)=x
\]
and, for all $4\le i\le k$,
\[
f(x,\ldots,x,y,x,\ldots,x)=x \mbox{ where } y \mbox{ is in position } i.
\]
\end{itemize}
An
identity $t_1=t_2$ is said to be \emph{linear} if both $t_1$ and $t_2$ involve
at most one occurrence of an operation symbol (e.g., $f(x,y)=g(x)$, or
$h(x,y,x)=x$).
An identity $t_1=t_2$ is said to be \emph{balanced}\footnote{
This notion of balanced identity is not related to the balanced digraphs introduced
in Section~\ref{sec:reduction}.
}
if the set of variables occurring in $t_1$ and $t_2$ are the same.
For example, both $f(x,x,y)=g(y,y,x)$ and $f(x,x,x)=x$ are balanced identities.
A set $\Sigma$ of identities is linear if it only contains linear identities,
idempotent if for each operation symbol, $f$, the identity $f(x,x,...,x)=x$ is
in $\Sigma$ and balanced if all of the identities in $\Sigma$ are balanced. 

Note that the identities defining WNU, symmetric and cyclic operations above are linear
and balanced. The identities defining edge operations, on the other hand, are
linear but \emph{not} balanced.

We now give some examples of results about the VCSP that are described using identities. 
We start with the notion of rigid cores~\cite{Kozik15:icalp}.
\begin{definition} 
A valued constraint language $\Gamma$ is a \emph{rigid core} if the only unary operation in
$\supp(\Gamma)$ is the identity operation. 
\end{definition}
It is known that with respect to tractability it suffices to consider valued
constraint languages that are \emph{rigid cores}. Indeed, for every valued
constraint language $\Gamma$ which is not a rigid core there is another language
$\Gamma'$ which is a rigid core\footnote{ $\Gamma'$ is the restriction of
$\Gamma$ to a subset $D'$ of the domain of $\Gamma$ together with the unary
relations $u_d$ for every $d\in D'$, where $u_d$ is defined by $u_d(d)=0$ and
$u_d(x)=\infty$ if $x\neq d$.} and with the property that $\VCSP(\Gamma)$ is
polynomial-time equivalent to
$\VCSP(\Gamma')$~\cite{Kozik15:icalp}.
It is also known that $\Gamma$ is a rigid core if and only if 
all operations from $\supp(\Gamma)$ are idempotent~\cite{Kozik15:icalp}.

The ``algebraic dichotomy conjecture''~\cite{Bulatov05:classifying}, 
a refinement of the dichotomy
conjecture for the CSP~\cite{Feder98:monotone}, can be re-stated as
follows~\cite{Bulatov05:classifying,MarotiMcK08:weakNU}: for a rigid core crisp language
$\Gamma$, $\CSP(\Gamma)$ is tractable if $\Gamma$ admits a WNU polymorphism of
some arity, and $\NP$-complete otherwise. Equivalently, $\CSP(\Gamma)$ is
tractable if $\Gamma$ admits a cyclic polymorphism of some arity, and is
$\NP$-complete otherwise.

The ``bounded-width theorem'' for the CSP can be restated as
follows~\cite{Larose07:bounded,Barto14:jacm,Bul09:boundedWidth,Bulatov16:lics}: 
for a rigid core crisp language
$\Gamma$, the problem $\CSP(\Gamma)$ has bounded width (and thus can be solved
using local consistency methods) if and only if $\Gamma$ has WNU polymorphisms
of all arities. 

There is an algorithmic technique for the CSP that generalises the idea of using
Gaussian elimination to solve simultaneous linear equations. The most general
version of this approach is based on the property of having a polynomial-sized
representation for the solution set of any
instance~\cite{Bulatov06:maltsev,Idziak10:siam}. This algorithm is called the
``few subpowers'' algorithm (because it is related to a certain algebraic
property to do with the number of subalgebras in powers of an algebra). Crisp
languages where this algorithm is guaranteed to find a solution (or show that
none exists) were captured in~\cite{Idziak10:siam}: for a crisp language
$\Gamma$, the problems $\CSP(\Gamma)$ are solvable using the few subpowers
algorithm if $\Gamma$ admits an edge polymorphism of some arity.
In fact the converse to this theorem is true in the following sense: the absence
of edge polymorphisms of $\Gamma$ implies that the presence of small enough
representations is not guaranteed~\cite{Idziak10:siam}.

For finite-valued constraint languages, the following complexity classification
has been obtained~\cite{tz16:jacm}: for a finite-valued constraint language $\Gamma$,
$\VCSP(\Gamma)$ is tractable if $\supp(\Gamma)$ contains a binary symmetric operation, and is
$\NP$-complete otherwise.

The power of the basic linear programming relaxation has been characterised as
follows~\cite{ktz15:sicomp}: for a valued constraint language $\Gamma$, the
problem $\VCSP(\Gamma)$ is solvable optimally by the basic linear programming
relaxation if and only if $\supp(\Gamma)$ contains symmetric operations of all
arities.

The power of constant-level Sherali-Adams linear programming relaxations has
been characterised as follows~\cite{tz17:sicomp}: for a valued
constraint language $\Gamma$, the problem $\VCSP(\Gamma)$ is solvable
optimally by a constant-level of the Sherali-Adams linear programming
relaxation~\cite{Sherali1990} if and only if the problem $\VCSP(\Gamma)$ is
solvable optimally by the \emph{third} level of the Sherali-Adams linear programming
relaxation if and only if $\supp(\Gamma)$ contains WNU operations of all arities.

\section{Reduction to a Single Combined Cost Function}

Throughout this paper we will make use of the following 
simple but useful observation about arbitrary finite languages.

\begin{proposition}\label{prop:single}
For any valued constraint language $\Gamma$ such that $|\Gamma|$ is finite,
there is a single cost function $\phi_{\Gamma}$ over the same domain 
such that:
\begin{enumerate}
\item $\pol(\Gamma) = \pol(\{\phi_{\Gamma}\})$;
\item $\fpol(\Gamma) = \fpol(\{\phi_{\Gamma}\})$;
\item $\mbox{VCSP}(\Gamma)$ and $\mbox{VCSP}(\{\phi_{\Gamma}\})$ are polynomial-time equivalent.\end{enumerate}
\end{proposition}
\begin{proof}
Let $\Gamma$ consist of $q$ cost functions,
$\phi_1,\ldots,\phi_q$, with arities $m_1,\ldots,m_q$, respectively.
Without loss of generality, we assume that none of the $\phi_i$ are
the constant function $\infty$. Let $m=\sum_{i=1}^q m_i$. Define the
cost function $\phi_\Gamma$, with arity $m$, by setting
$\phi_\Gamma(x_1,\ldots,x_m)=\phi_1(x_1,\ldots,x_{m_1})+\phi_2(x_{{m_1}+1},\ldots,x_{m_{1}+m_{2}})+\ldots+\phi_q(x_{m-m_{q}+1},\ldots,x_m)$,
and set $\GammaC = \{\phi_\Gamma\}$.

Since the operations in $\pol(\Gamma)$ are applied co-ordinatewise, 
it follows easily from Definition~\ref{def:polymorphism}
that $\pol(\Gamma) = \pol(\{\phi_{\Gamma}\})$, 
and since inequalities are preserved by addition,
it follows easily from Definition~\ref{def:fpol} 
that $\fpol(\Gamma) = \fpol(\{\phi_{\Gamma}\})$.

For any instance $\I$ of $\mbox{VCSP}(\Gamma)$ we can obtain an equivalent
instance $\I'$ of $\mbox{VCSP}(\{\phi_\Gamma\})$ by simply adding irrelevant variables to
the scope of each constraint $\phi_i(\vec{x})$,
which are constrained by the elements of $\Gamma\setminus \{\phi_i\}$,
and then minimising over these.
The assignments that minimise the objective function of $\I$ can then be obtained by
taking the assignments that minimise the objective function of $\I'$ and restricting them
to the variables of $\I$.

Conversely, any instance $\I'$ of $\mbox{VCSP}(\{\phi_\Gamma\})$ can clearly be expressed
as an instance of $\mbox{VCSP}(\Gamma)$
since each constraint in
$\I'$ can be expressed as a sum of constraints whose constraint
cost functions are contained in $\Gamma$.
\end{proof}

\section{Reduction by the Dual Encoding}
\label{sec:dual}

In this section we will describe the dual encoding introduced
in~\cite{Dechter89:tree} for the CSP and later extended in~\cite{Dechter00:dual}
to the VCSP.

\subsection{From a language $\Gamma$ to a binary language $\GammaD$}

\begin{definition}
\label{def:gammadual}
Let $\Gamma$ be any valued constraint language over $D$, such that $|\Gamma|$ is finite,
and let $\phi_\Gamma$ be the corresponding single cost function, of arity $m$, 
as defined in Proposition~\ref{prop:single}.

The \emph{dual} of $\Gamma$, denoted $\GammaD$, is the binary valued constraint language
with domain $D'=\Feas(\phi_\Gamma)\subseteq D^m$, defined by
\[
\GammaD\ =\ \{\phi'_\Gamma\}\ \cup\ \bigcup_{i,j\in\{1,\ldots,m\}}\ \{match_{i,j}\}\,,
\]
where $\phi'_\Gamma:D'\to\Q$ is the unary finite-valued cost function on $D'$ defined by
$\phi'_\Gamma(\vec{x})=\phi_\Gamma(x_1,\ldots,x_m)$
for every $\vec{x} = (x_1,\ldots,x_m) \in D'$,
and each $match_{i,j}:D'\times D'\to\qq$ is the binary relation on $D'$ defined by
\[
match_{i,j}(\vec{x},\vec{y})
\ =\ \begin{cases}
0 & \mbox{if $\vec{x}[i]=\vec{y}[j]$} \\
\infty & \mbox{otherwise}.
\end{cases}
\]
\end{definition}
The language $\GammaD$ contains a single unary cost function,
which returns only finite values,
together with $m^2$ binary relations and hence is a binary 
valued constraint 
language.

\begin{example}
\label{ex:eq}
Let $\Gamma=\{\phi_{eq}\}$, where $\phi_{eq}$ is the equality relation on $D$, i.e.,
$\phi_{eq}:D\times D\to\qq$ is defined by $\phi_{eq}(x,y)=0$ if $x=y$ and
$\phi_{eq}(x,y)=\infty$ if $x\neq y$.

Then $D'=\Feas(\phi_{eq})=\{(a,a)\:|\:a\in D\}$ and $\GammaD$
consists of a single unary finite-valued cost function $\phi'_{eq}$, together with
four binary relations $match_{1,1}, match_{1,2}, match_{2,1}$, and
$match_{2,2}$.

Moreover, $\phi'_{eq}(\vec{x})=0$ for every $\vec{x}\in D'$, and hence is trivial.
All four of the other relations are in fact
equal to the equality relation on $D'$ defined by $\{((a,a),(a,a))\:|\:(a,a)\in D'\}$.
Thus, the dual of the equality relation on $D$ consists of a trivial unary relation,
together with the equality relation on $D'$, where $|D|=|D'|$.
\end{example}

\begin{example}
\label{ex:onein3}
Let $\Gamma=\{\phi_{sum}\}$, where $\phi_{sum}: \{0,1\}^3 \rightarrow \qq$ is defined as follows:
\[
\phi_{sum}(x,y,z)
\ =\ \begin{cases}
x + 2y + 3z & \mbox{if $x + y + z = 1$} \\
\infty & \mbox{otherwise}.
\end{cases}
\]
Then $D'=\Feas(\phi_{sum})=\{(1,0,0),(0,1,0),(0,0,1)\}$ and $\GammaD$
consists of a single unary finite-valued cost function $\phi'_{sum}$, together with
nine binary relations $match_{1,1}, match_{1,2},$ $match_{1,3}, \ldots, match_{3,3}$.

If we set $\vec{a} = (1,0,0), \vec{b} = (0,1,0), \vec{c} = (0,0,1)$, then we have
$\phi'_{sum}(\vec{a}) = 1; \phi'_{sum}(\vec{b}) = 2$ and $\phi'_{sum}(\vec{c}) = 3$.
Also
\[
match_{1,1}(\vec{x},\vec{y}) \,=\,\begin{cases}
0 & \mbox{if $(\vec{x},\vec{y}) \in
\{(\vec{a},\vec{a}),(\vec{b},\vec{b}),(\vec{b},\vec{c}),(\vec{c},\vec{b}),(\vec{c},\vec{c})\}$} \\
\infty & \mbox{otherwise}
\end{cases}
\]
\[
match_{1,2}(\vec{x},\vec{y})=\begin{cases}
0 & \mbox{if $(\vec{x},\vec{y}) \in
\{(\vec{a},\vec{b}),(\vec{b},\vec{a}),(\vec{b},\vec{c}),(\vec{c},\vec{a}),(\vec{c},\vec{c})\}$}\\
\infty & \mbox{otherwise}
\end{cases}
\]
and so on.
\end{example}

\subsection{The dual encoding using $\GammaD$} 

We will need the following notation: for any $\vec{x}_i\in X^m$ with
$\vec{x}_i=(x_{i_1},\ldots,x_{i_m})$, we write $vars(\vec{x}_i)$ for the set
$\{x_{i_1},\ldots,x_{i_m}\}$.

\begin{definition}
\label{def:Id}
Let $\Gamma$ be any valued constraint language over $D$, such that $|\Gamma|$ is finite,
and let $\phi_\Gamma$ be the corresponding single cost function, of arity $m$, 
as defined in Proposition~\ref{prop:single}.
Let $\I$ be an arbitrary instance of $\VCSP(\{\phi_\Gamma\})$ 
with variables $X=\{x_1,\ldots,x_n\}$,
domain $D$, and constraints $\phi_\Gamma(\vec{x}_1),\ldots,\phi_\Gamma(\vec{x}_q)$, where
$\vec{x}_i\in X^m$ for all $1\leq i\leq q$. 

The \emph{dual} of $\I$, denoted $\ID$, is defined to be 
the following instance of $\VCSP(\GammaD)$:
\begin{itemize}
\item The variables
$V'=\{x'_1,\ldots,x'_q\}$ of $\ID$ are the constraints of $\I$.
\item The domain of $\ID$ is $D'=\Feas(\phi_\Gamma)\subseteq D^m$.
\item For every $1\leq i\leq q$, there is a unary constraint
$\phi'_\Gamma(x'_i)$,
where $\phi'_\Gamma:D'\to\Q$ is as defined in Definition~\ref{def:gammadual}.
\item If the scopes of two constraints of $\I$, say $\phi_\Gamma(\vec{x}_i)$ and
$\phi_\Gamma(\vec{x}_j)$,
overlap, then there are binary constraints between $x'_i$ and $x'_j$ enforcing
equality at the overlapping coordinate positions. More specifically, if
$\vec{x}_i=(x_{i_1},\ldots,x_{i_m})$,
$\vec{x}_j=(x_{j_1},\ldots,x_{j_m})$, and $vars(\vec{x}_i)\cap
vars(\vec{x}_j)\neq\emptyset$ then there is a binary
constraint $match_{k,l}(x'_i,x'_j)$ for every $k,l\in\{1,\ldots,m\}$ with
$i_k=j_l$.

\end{itemize}
\end{definition}

The dual encoding provides a way to reduce instances of $\VCSP(\Gamma)$ to
instances of $\VCSP(\GammaD)$. Our next result extends this observation to obtain the
reverse reduction as well.
\begin{proposition}
\label{prop:dualequiv}
For any valued constraint language $\Gamma$ such that $|\Gamma|$ is
finite, if $\GammaD$ is the dual of $\Gamma$, 
then $\mbox{VCSP}(\Gamma)$ and $\mbox{VCSP}(\GammaD)$ are polynomial-time equivalent.
\end{proposition}
\begin{proof}
By Proposition~\ref{prop:single} we may assume that $\Gamma$ consists of a single
cost function $\phi_\Gamma: D^m \rightarrow \qq$. Moreover, since $D$ is finite,
and $m$ is fixed,
we may assume that this cost function is given extensionally as a table of values.

Hence, for any instance $\I$ of $\VCSP(\Gamma)$ we can construct in polynomial time
the dual instance $\ID$ in $\VCSP(\GammaD)$, as defined above (Definition~\ref{def:Id}).
It is straightforward to show that
the assignments that minimise the objective function of $\ID$ correspond precisely
to the assignments that minimise the objective function of $\I$, and hence we have a
polynomial-time reduction from $\VCSP(\Gamma)$ to $\VCSP(\GammaD)$.

For the other direction, given any instance $\I'$ in $\VCSP(\GammaD)$ we now
indicate how to construct a corresponding instance $\I$ in $\VCSP(\Gamma)$.

For each variable $x'_i$ of $\I'$ we introduce a fresh set of $m$
variables for $\I$. If there is a unary constraint
$\phi'_\Gamma(x'_i) \in \I'$, then we introduce the constraint
$\phi_\Gamma$ on the corresponding variables of $\I$. If there is no
unary constraint on $x'_i$, then we introduce the constraint
$\Feas(\phi_\Gamma)$ on the corresponding variables of $\I$ to code
the fact that the domain of $x'_i$ is $D'$. If there is a binary
constraint $match_{k,l}(x'_i,x'_j)$ in $\I'$, then we merge the $k$th
and $l$th variables in the corresponding sets of variables in $\I$.
This construction can be carried out in polynomial time.

We have constructed an instance $\I$ in $\VCSP(\{\phi_\Gamma,\Feas(\phi_\Gamma)\}$ such that
assignments minimising the objective function of $\I$ correspond precisely to
assignments minimising the objective function of $\I'$. Hence we have established a 
polynomial-time reduction from $\VCSP(\GammaD)$ to $\VCSP(\Gamma \cup \{\Feas(\phi_\Gamma)\})$.

However, it follows from the proof of~\cite[Theorem~4.3]{cccjz13:sicomp} that
$\VCSP(\Gamma \cup \{\Feas(\phi_\Gamma)\})$ can be reduced to $\VCSP(\Gamma)$ in polynomial time.
\end{proof}

\subsection{Preservation of algebraic properties}

Our next result shows that the  polymorphisms of $\GammaD$ are very closely related
to the polymorphisms of $\Gamma$.
\begin{theorem}
\label{the:samepol}
Let $\Gamma$ be a valued constraint language such that $|\Gamma|$ is finite,
and let $\GammaD$ be the dual of $\Gamma$.
There is a one-to-one correspondence between the polymorphisms of $\Gamma$ and the 
polymorphisms of $\GammaD$, defined as follows.
For any $f \in \polk{\Gamma}$ the corresponding operation $\fD \in \polk{\GammaD}$
is defined by $\fD(\vec{x}_1,\ldots,\vec{x}_k) =
\extend{f}(\vec{x}_1,\ldots,\vec{x}_k)$ for all $\vec{x}_i$
in the domain of $\GammaD$.
\end{theorem}
\begin{proof}
By Proposition~\ref{prop:single} we may assume that $\Gamma$ consists of a single
cost function $\phi_\Gamma: D^m \rightarrow \qq$, and hence that
the domain $D'$ of $\GammaD$ is a subset of $D^m$.

First, consider any $f:D^k \rightarrow D \in \polk{\Gamma}$, and the corresponding
$\fD:(D')^k \rightarrow D'$ given by $\fD(\vec{x}_1,\ldots,\vec{x}_k) =
\extend{f}(\vec{x}_1,\ldots,\vec{x}_k)$ for all $\vec{x}_i \in D'$.
Since $f$ is a polymorphism of $\phi_\Gamma$,
it is also a polymorphism of the unary cost function $\phi'_\Gamma$ in $\GammaD$.
It is straightforward to check that $\fD$ is also a polymorphism of all binary $match_{i,j}$ relations in $\GammaD$ (since it will return the same label at all positions where
its arguments have the same label). Hence $\fD \in \polk{\GammaD}$.

Now consider any $\fD:(D')^k \rightarrow D' \in \polk{\GammaD}$. Since $\fD$ is a
polymorphism of $match_{i,i}$ it must return an element of $D'$ whose label in position $i$
is a function, $g_i$, of the labels in position $i$ of its arguments. Moreover, since $\fD$ is a
polymorphism of $match_{i,j}$, the functions $g_i$ and $g_j$ must return the same results
for all possible arguments from $D'$.
Hence, there is a single function $g:D^k \rightarrow D$ such that the result returned by
$\fD(\vec{x}_1,\ldots,\vec{x}_k)$ is equal to $\extend{g}(\vec{x}_1,\ldots,\vec{x}_k)$.
Now, since $\fD$ must return an element of $D'$,
it follows that $g$ must be a polymorphism of $\phi_\Gamma$, which gives the result.
\end{proof}
The individual cost functions in $\GammaD$ often have other polymorphisms,
that are not of the form indicated in Theorem~\ref{the:samepol}, but the only
polymorphisms that are shared by every cost function in $\GammaD$ are those
that correspond to polymorphisms of $\Gamma$ in this way,
as the next example illustrates.
\begin{example}
Recall the language $\Gamma=\{\phi_{sum}\}$, defined
in Example~\ref{ex:onein3}.

The cost function $\phi_{sum}$ has no polymorphisms, except for the projection
operations on $D = \{0,1\}$.

However, the unary finite-valued cost function $\phi'_{sum}$, has \emph{every} operation on
$D' = \{\vec{a},\vec{b},\vec{c}\}$ as a polymorphism.

The binary relation $match_{1,1}$ has many operations on $D'$ as polymorphisms,
including all of the constant operations.

The binary relation $match_{1,2}$ also has many operations on $D'$ as polymorphisms,
including the ternary \emph{majority} operation $g$ defined by
\[
g(\vec{x},\vec{y},\vec{z}) \,=\,\begin{cases}
\vec{x} & \mbox{if } $\vec{x} = \vec{y}$ \mbox{ or } $\vec{x} = \vec{z}$ \\
\vec{y} & \mbox{if } $\vec{y} = \vec{z}$ \\
\vec{c} & \mbox{otherwise}
\end{cases}
\]
but not including the constant operation returning the label $\vec{a}$,
or the constant operation returning the label $\vec{b}$.

Continuing in this way it can be shown that 
the only operations that are polymorphisms of \emph{every} cost function in $\GammaD$
are the projection operations on $D'$.
\end{example}

One simple consequence of Theorem~\ref{the:samepol} is that the polymorphisms of
$\Gamma$ and the polymorphisms of $\GammaD$ satisfy exactly the same identities.
\begin{corollary}
\label{cor:sameident}
Let $\Gamma$ be a valued constraint language such that $|\Gamma|$ is finite,
and let $\GammaD$ be the dual of $\Gamma$.
Then the operations in $\pol(\Gamma)$ and the operations in $\pol(\GammaD)$ satisfy
exactly the same identities.
\end{corollary}

\begin{corollary}
Let $\Gamma$ be a valued constraint language such that $|\Gamma|$ is finite,
and let $\GammaD$ be the dual of $\Gamma$.
Then $\Gamma$ is a rigid core if and only if $\GammaD$ is a rigid core.
\end{corollary}
\begin{proof}
Follows immediately from Corollary~\ref{cor:sameident}, since the property of being
idempotent is specified by an identity, as discussed in
Section~\ref{sec:identities}.
\end{proof}

Following our discussion in Section~\ref{sec:identities},
Corollary~\ref{cor:sameident} shows that the property of being solvable using
local consistency methods or by the few subpowers algorithm is possessed by a
language $\Gamma$ if and only if it is also possessed by the associated binary
language $\GammaD$.

Although the polymorphisms of $\Gamma$ and $\GammaD$ satisfy the same identities,
the polymorphisms of $\GammaD$ do not, in general, have \emph{all} the same
properties as the polymorphisms of $\Gamma$. For example, $\pol(\Gamma)$ might include
the binary operation $\min$ that returns the smaller of its two arguments, according to some
fixed ordering of $D$. This operation has the property of being \emph{conservative},
which means that the result is always equal to one of the arguments. However, the
corresponding operation $\min_d$ in $\pol(\GammaD)$ is \emph{not} generally conservative,
since, for example, $\min_d((a,b),(b,a)) = (a,a)$ for all $a<b$.

Our next result shows that the fractional polymorphisms of $\GammaD$ are 
closely related
to the fractional polymorphisms of $\Gamma$.
\begin{theorem}
\label{thm:samefpol}
Let $\Gamma$ be a valued constraint language such that $|\Gamma|$ is finite,
and let $\GammaD$ be the dual of $\Gamma$.
There is a one-to-one correspondence between the fractional polymorphisms of $\Gamma$
and the fractional polymorphisms of $\GammaD$, defined as follows.
For any $\omega:\polk{\Gamma} \to\qnn \in \fpolk{\Gamma}$ the corresponding 
function $\omega_d:\polk{\GammaD} \to\qnn \in \fpolk{\GammaD}$
is defined by $\omega_d(\fD) = \omega(f)$ for all $f \in \polk{\Gamma}$ 
and their corresponding operations $\fD \in
\polk{\GammaD}$ (as defined in Theorem~\ref{the:samepol}).
\end{theorem}
\begin{proof}
By Proposition~\ref{prop:single} we may assume that $\Gamma$ consists of a single
cost function $\phi_\Gamma: D^m \rightarrow \qq$, and hence that
the domain $D'$ of $\GammaD$ is a subset of $D^m$.

First, consider any $\omega:\polk{\Gamma} \rightarrow \qnn \in
\fpolk{\Gamma}$, and the corresponding $\omega_d:\polk{\GammaD}
\rightarrow \qnn$ given by $\omega_d(\fD) = \omega(f)$ for all $f \in
\polk{\Gamma}$. Since $\omega$ is a fractional polymorphism of
$\phi_\Gamma$, it is easy to check that $\omega_d$ satisfies the
conditions in Definition~\ref{def:fpol}, and hence is a
fractional polymorphism of the unary cost function $\phi'_\Gamma$
in $\GammaD$. Since all other cost functions in $\GammaD$ are
the $match_{i,j}$ relations, the inequality condition in
Definition~\ref{def:fpol} holds trivially for all these cost 
functions, and hence $\omega_d$ is a fractional polymorphism of all
cost functions in $\GammaD$.

Now consider any $\omega_d:\polk{\GammaD} \rightarrow \qnn \in
\fpolk{\GammaD}$. Since $\omega_d$ is a fractional polymorphism of
$\phi'_\Gamma$, the function $\omega: \polk{\Gamma} \to \qnn$ that
assigns the same weights to corresponding elements of $\polk{\Gamma}$
satisfies the condition of Definition~\ref{def:fpol}, and hence
is a fractional polymorphism of $\phi_\Gamma$.
\end{proof}

Following our discussion in Section~\ref{sec:identities}, combining
Corollary~\ref{cor:sameident} with Theorem~\ref{thm:samefpol} shows that the
property of being solvable using the basic linear programming relaxation or by
constant levels of the Sherali-Adams linear programming relaxations is possessed by
a language $\Gamma$ if and only if it is also possessed by the associated binary
language $\GammaD$.

\section{Reduction by the Extended Dual Encoding}
\label{sec:reduction}

In this section we will describe our new extension of the reduction
from~\cite{Bulin15:lmcs} to the VCSP.

\subsection{From a language $\Gamma$ to a binary language $\GammaE$}

Throughout this section it will be helpful to view a binary relation on a set 
as a directed graph (digraph) 
where the vertices are the elements of the set, 
and the directed edges are the binary tuples in the relation.

First we introduce some simple definitions relating to digraphs that we will need
in our constructions. 
We define a digraph as a structure $\G=(V^{\G},E^{\G})$ with vertices $v\in V^{\G}$ and
directed edges $e\in E^{\G}$. 
We will sometimes write the directed edge $(a,b)\in E^{\G}$ as $a\to b$.

\begin{definition}
A digraph is an \emph{oriented path} if it consists of 
a sequence of vertices $v_0,v_1,...v_k$ 
such that precisely one of $(v_{i-1},v_i),(v_i,v_{i-1})$ is an edge, for each $i=1,...,k$. 
\end{definition}

We now adapt the construction from~\cite{Bulin15:lmcs} to valued constraint languages.
The construction makes use of \emph{zigzags}, where a zigzag is 
the oriented path $\bullet\to\bullet\leftarrow\bullet\to\bullet$.
The important property we will use is that there is a surjective homomorphism
from a zigzag to a single edge but not from a single edge to a zigzag.
\begin{definition}
\label{def:gammaextdual}
Let $\Gamma$ be any valued constraint language over $D$, such that $|\Gamma|$ is finite,
and let $\phi_\Gamma$ be the corresponding single
cost function, of arity $m$, as defined in 
Proposition~\ref{prop:single}. As before, we define
$D'=\Feas(\phi_\Gamma)\subseteq D^m$.

The \emph{extended dual} of $\Gamma$, denoted $\GammaE$, 
is the binary valued constraint language $\{\DG,\mu_\Gamma\}$,
where $\DG$ is a binary relation, and $\mu_\Gamma$ is a unary cost function, as
defined below.

For $S\subseteq \{1,2\zd m\}$ define $\Q_{S,i}$ to be a single edge if $i\in S$, and a zigzag if $i\in \{1,2\zd m\}\setminus S$.
Now define the oriented path $\Q_S$ 
by 
\[
\Q_{S}=\bullet\to\bullet\; \dot{+}\; \Q_{S,1} \;\dot{+}\; \Q_{S,2} \;\dot{+}\;\dotsb\;\dot{+}\; \Q_{S,m} \;\dot{+}\;\bullet\to\bullet
\]
where $\dot{+}$ denotes the concatenation of paths.

To define the digraph $\DG$, consider the binary relation $D \times D'$ as a digraph,
and replace each edge $(d,\vec{x})$ with the oriented path $\Q_{\{i \;\mid\; \vec{x}[i]=d\}}$. 
The resulting digraph $\DG$ has vertex set $V^{\DG} = D \cup D' \cup E$, 
where $E$ consists of all the additional internal vertices from the oriented paths $\Q_S$.

Finally, let $\mu_\Gamma$ be the unary cost function on $V^{\DG}$ such that
\[
\mu_\Gamma(v)=\left\{\begin{array}{l l}\phi_\Gamma(v) & \quad\text{if }v\in D'\\
                                    0 & \quad \text{otherwise.}\end{array}\right.
\]
\end{definition}
The language $\GammaE$ contains a single binary relation $\DG$,
together with a unary cost function $\mu_\Gamma$,
which returns only finite values,
and hence is a binary valued constraint language with domain $V^{\DG}$.

\begin{example}\label{ex1}
Consider the valued constraint language $\Gamma$ over the domain $D=\{0,1\}$ 
containing the single (binary) cost function 
\[
\rho(x,y)=\left\{\begin{array}{l l}2 & \quad\text{if }(x,y)=(0,1)\\
                                   1 & \quad \text{if }(x,y)=(1,0)\\
                                   \infty & \quad \text{otherwise.}
                \end{array}\right.
\]
The digraph $\DG$ constructed from $\rho$ is shown in Figure~\ref{digraphEx1}.
The unary cost function built from $\rho$ is 
\[
\mu_\Gamma(v)=\left\{\begin{array}{l l}2 & \quad\text{if }v=(0,1)\\
                              1 & \quad \text{if }v=(1,0)\\
                              0 & \quad \text{otherwise}
            \end{array}\right.
\] 
for every vertex $v\in V^{\DG}$.
\end{example}

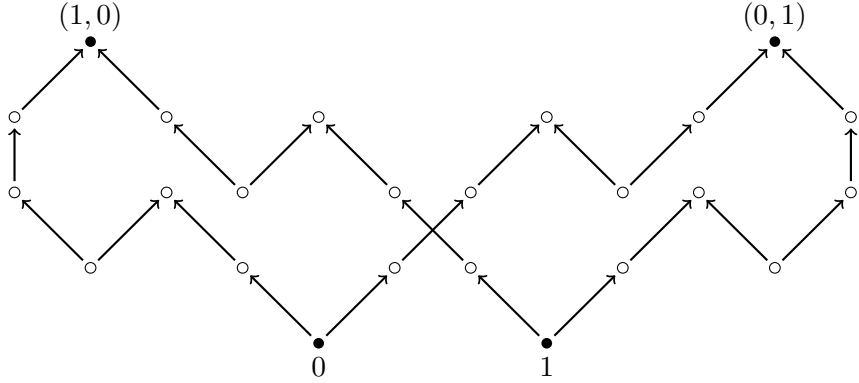
\begin{figure}\label{digraphEx1}
\begin{center}
\begin{tikzpicture}

\fill (canvas cs:x=2cm,y=5cm) circle (2pt);
\fill (canvas cs:x=11cm,y=5cm) circle (2pt);
\fill (canvas cs:x=5cm,y=1cm) circle (2pt);
\fill (canvas cs:x=8cm,y=1cm) circle (2pt);

\draw (canvas cs:x=1cm,y=4cm) circle (2pt);
\draw (canvas cs:x=3cm,y=4cm) circle (2pt);
\draw (canvas cs:x=5cm,y=4cm) circle (2pt);
\draw (canvas cs:x=8cm,y=4cm) circle (2pt);
\draw (canvas cs:x=10cm,y=4cm) circle (2pt);
\draw (canvas cs:x=12cm,y=4cm) circle (2pt);

\draw (canvas cs:x=1cm,y=3cm) circle (2pt);
\draw (canvas cs:x=3cm,y=3cm) circle (2pt);
\draw (canvas cs:x=4cm,y=3cm) circle (2pt);
\draw (canvas cs:x=6cm,y=3cm) circle (2pt);
\draw (canvas cs:x=7cm,y=3cm) circle (2pt);
\draw (canvas cs:x=9cm,y=3cm) circle (2pt);
\draw (canvas cs:x=10cm,y=3cm) circle (2pt);
\draw (canvas cs:x=12cm,y=3cm) circle (2pt);

\draw (canvas cs:x=2cm,y=2cm) circle (2pt);
\draw (canvas cs:x=4cm,y=2cm) circle (2pt);
\draw (canvas cs:x=6cm,y=2cm) circle (2pt);
\draw (canvas cs:x=7cm,y=2cm) circle (2pt);
\draw (canvas cs:x=9cm,y=2cm) circle (2pt);
\draw (canvas cs:x=11cm,y=2cm) circle (2pt);

Level 0 vertex labels
\node [below] at (5,0.95) {$0$};
\node [below] at (8,0.95) {$1$};
Top level vertex labels
\node [above] at (2,5) {$(1,0)$};
\node [above] at (11,5) {$(0,1)$};

Path 0 -> (1,0)
\draw [->,thick](1.1,4.1) -- (1.9,4.9);
\draw [->,thick](1,3.15) -- (1,3.85);
\draw [->,thick](1.9,2.1) -- (1.1,2.9);
\draw [->,thick](2.1,2.1) -- (2.9,2.9);
\draw [->,thick](3.9,2.1) -- (3.1,2.9);
\draw [->,thick](4.9,1.1) -- (4.1,1.9);

Path 0 -> (0,1)
\draw [->,thick](10.1,4.1) -- (10.9,4.9);
\draw [->,thick](9.1,3.1) -- (9.9,3.9);
\draw [->,thick](8.9,3.1) -- (8.1,3.9);
\draw [->,thick](7.1,3.1) -- (7.9,3.9);
\draw [->,thick](6.1,2.1) -- (6.9,2.9);
\draw [->,thick](5.1,1.1) -- (5.9,1.9);

Path 1 -> (1,0)
\draw [->,thick](2.9,4.1) -- (2.1,4.9);
\draw [->,thick](3.9,3.1) -- (3.1,3.9);
\draw [->,thick](4.1,3.1) -- (4.9,3.9);
\draw [->,thick](5.9,3.1) -- (5.1,3.9);
\draw [->,thick](6.9,2.1) -- (6.1,2.9);
\draw [->,thick](7.9,1.1) -- (7.1,1.9);

Path 1 -> (0,1)
\draw [->,thick](11.9,4.1) -- (11.1,4.9);
\draw [->,thick](12,3.15) -- (12,3.85);
\draw [->,thick](11.1,2.1) -- (11.9,2.9);
\draw [->,thick](10.9,2.1) -- (10.1,2.9);
\draw [->,thick](9.1,2.1) -- (9.9,2.9);
\draw [->,thick](8.1,1.1) -- (8.9,1.9);

\end{tikzpicture}

\caption{The digraph $\DG$ built from the valued constraint language $\Gamma$
described in Example~\ref{ex1}.}
\end{center}
\end{figure}

The binary relation $\DG$ defined in Definition~\ref{def:gammaextdual}
is identical to the digraph defined in~\cite[Definition~3.2]{Bulin15:lmcs}, 
where it is shown that 
the number of vertices in $\DG$ is $(3n+1)|D'||D|+(1-2n)|D'|+|D|$ 
and the number of edges is $(3n+2)|D'||D|-2n|D'|$.
Also, as noted in~\cite{Bulin15:lmcs}, this construction can be performed in polynomial time.

\subsection{The extended dual encoding using $\GammaE$}

We now show how to reduce instances of $\VCSP(\Gamma)$ to
instances of $\VCSP(\GammaE)$ using a construction that we call the extended dual encoding.
This construction is similar in overall structure to the hidden variable encoding 
described in~\cite{Rossi90:equivalence}, but has only one form of binary constraint.

\begin{definition}
\label{def:extdual}
Let $\Gamma$ be any valued constraint language over $D$, such that $|\Gamma|$ is finite,
and let $\phi_\Gamma$ be the corresponding single cost function, of arity $m$, 
as defined in Proposition~\ref{prop:single}.
Let $\I$ be an arbitrary instance of $\VCSP(\{\phi_\Gamma\})$ 
with variables $X=\{x_1,\ldots,x_n\}$,
domain $D$, and constraints $\phi_\Gamma(\vec{x}_1),\ldots,\phi_\Gamma(\vec{x}_q)$, where
$\vec{x}_i\in X^m$ for all $1\leq i\leq q$. 

The \emph{extended dual} of $\I$, denoted $\IE$, is defined to be 
the following instance of $\VCSP(\GammaE)$:

\begin{itemize}
\item The variables of $\IE$ are $X \cup \{x'_1,\ldots,x'_q\} \cup Y$
where $\{x'_1,\ldots,x'_q\}$ correspond to the constraints of $\I$
and $Y$ contains additional variables as described below.
\item The domain of $\IE$ is the same as the domain of $\GammaE$,
as defined in Definition~\ref{def:gammaextdual};
that is, $D \cup D' \cup E$, where $D'=\Feas(\phi_\Gamma)\subseteq D^m$
and $E$ contains the additional vertices of $\DG$.
\item For every $1\leq i\leq q$, there is a unary constraint
$\mu_\Gamma(x'_i)$,
where $\mu_\Gamma$ is as defined in Definition~\ref{def:gammaextdual}.
\item 
For each constraint $\phi_\Gamma(\vec{x}_i)$ of $\I$,
where $\vec{x}_i=(x_{i_1},\ldots,x_{i_m})$,
there is an oriented path $\Q_{\{j\}}$ 
from each $x_{i_j}$ to $x'_i$ 
(where $\Q_S$ for any $S \subseteq \{1,2,\ldots,m\}$ is as defined in 
Definition~\ref{def:gammaextdual}).
Each such path uses disjoint sets of intermediate vertices,  
and each oriented edge on these paths, say $(y,y')$, is the scope of a constraint in $\IE$
with relation $\DG$. 
The set $Y$ is the union of all such intermediate vertices over all such paths.
\end{itemize}
\end{definition}

To verify that the extended dual of $\I$ gives a reduction from $\VCSP(\Gamma)$
to $\VCSP(\GammaE)$ we introduce the following terminology.

Given any digraph $\G$, we can define an associated undirected graph $\G^*$ where
each directed edge of $\G$ is replaced by an undirected edge on the same pair of vertices. 
We will say that a digraph $\G$ is connected if $\G^*$ is connected,
and we will define the connected components of $\G$ to be the connected
components of $\G^*$. 

For any digraph $\G$, if $\G^*$ contains a cycle, then the corresponding set of directed
edges in $\G$ will be called an oriented cycle.
The length of an oriented cycle is defined as being the absolute
value of the difference between the number of edges oriented in one direction 
around the cycle and edges oriented in the opposite direction. 
A connected digraph $\G$ is said to be \emph{balanced} if all of its oriented cycles have zero
length~\cite{Feder98:monotone}. 

Note that the digraph $\DG$ described in Definition~\ref{def:gammaextdual} is balanced.
Moreover, the binary scopes of the extended dual instance $\IE$ constructed in 
Definition~\ref{def:extdual} also form a balanced digraph which we will call $\G_{\IE}$
(if $\G_{\IE}$ is not connected then we may consider each connected component separately).

The vertices of any balanced digraph $\G$ can be
organised into \emph{levels}, which are non-negative integers 
given by a function $lvl$ such that for every directed edge $(a,b)\in E^{\G}$,
$lvl(b)=lvl(a)+1$. The minimum level of $\G$ is 0, and the top level is 
called the \emph{height} of $\G$. 

Any feasible solution to $\IE$ must assign to each vertex $x$ in $\G_{\IE}$  
a label $d_x$ chosen from the vertices of $\DG$, 
which must be \emph{at the same level} as $x$.

Every variable $x_i \in X$ of $\IE$ is at level $0$ in $\G_{\IE}$, 
and so any feasible solution to $\IE$ must assign
to $x_i$ a label at level 0 in $\DG$, that is, an element $d(x_i)$ of $D$. 
Similarly, every variable $x'_j$  of $\IE$ is at level $m+2$ in $\G_{\IE}$, 
and so any feasible solution to $\IE$ must assign to $x'_j$ a label 
at level $m+2$ in $\DG$, that is, an element $d(x'_j)$ of $D' = \Feas(\phi_\Gamma)$. 

Every other variable $y$ of $\IE$ lies on an oriented path of the form $\Q_{\{k\}}$ 
from some $x_i$ to $x'_j$, and so any feasible solution to $\IE$
must assign to all variables on this oriented path a label on some fixed 
oriented path of the form $\Q_S$ in $\DG$.
By the construction of the oriented paths $\Q_S$ (see Definition~\ref{def:gammaextdual}),
the labels assigned to the variables in a path of the form $\Q_{\{k\}}$ 
must lie in an oriented path of the form $\Q_S$
for some set $S$ that contains the index $k$~\cite[Observation~3.1]{Bulin15:lmcs}.
By Definition~\ref{def:gammaextdual}, 
such a path exists in $\DG$ if and only if $d(x_i) = d(x'_j)[k]$.

Hence there is a one-to-one correspondence between feasible solutions to $\IE$ 
and feasible solutions to $\I$. The cost of each feasible solution to $\IE$ is 
determined by the sum of the values given by the cost function $\mu_\Gamma$ 
for the labels assigned to the variables $x'_i$,
and hence is equal to the cost of the corresponding solution to $\I$.
Hence the extended dual encoding specified in 
Definition~\ref{def:extdual} provides a way to reduce instances of $\VCSP(\Gamma)$ to
instances of $\VCSP(\GammaE)$. 

Our next result extends this observation to obtain the
reverse reduction as well.
\begin{theorem}
\label{thm:extdualequiv}
For any valued constraint language $\Gamma$ such that $|\Gamma|$ is
finite, if $\GammaE$ is the extended dual of $\Gamma$, 
then $\mbox{VCSP}(\Gamma)$ and $\mbox{VCSP}(\GammaE)$ are polynomial-time equivalent.
\end{theorem}
\begin{proof}
By Proposition~\ref{prop:single} we may assume that $\Gamma$ consists of a single
cost function $\phi_\Gamma: D^m \rightarrow \qq$. Moreover, since $D$ is finite,
and $m$ is fixed,
we may assume that this cost function is given extensionally as a table of values.

Hence, for any instance $\I$ of $\VCSP(\Gamma)$ we can construct in polynomial time
the extended dual instance $\IE$ in $\VCSP(\GammaE)$ 
as described in Definition~\ref{def:extdual}.
As we have just shown, 
the assignments that minimise the objective function of $\IE$ correspond precisely
to the assignments that minimise the objective function of $\I$, and hence we have a
polynomial-time reduction from $\VCSP(\Gamma)$ to $\VCSP(\GammaE)$.

For the other direction, given any instance $\I'$ in $\VCSP(\GammaE)$ we now
indicate how to solve it in polynomial time, 
or else construct in polynomial time a corresponding instance $\I'_d$ in
$\VCSP(\GammaD)$.
We can then appeal to Proposition~\ref{prop:dualequiv}.

Consider the digraph $\G$ formed by the binary scopes of $\I'$.
Since each connected component can be considered separately, we may assume that $\G$
is connected.
Moreover, if $\G$ is not balanced, then $\I'$ has no feasible solutions,
so we may assume that $\G$ is balanced (which can be checked in polynomial time).

Any feasible solution to $\I'$ must assign each vertex in $\G$ 
a label chosen from the vertices of $\DG$, in a way which preserves the 
differences in levels between different vertices.
Hence if the height of $\G$ is greater than the height of $\DG$, 
then $\I'$ has no feasible solutions,
so we may assume that the height of $\G$ is less than or equal to the height of $\DG$.

Now consider the case when $\G$ is balanced and of height $h$ 
which is strictly less than the height of $\DG$.
In this case every vertex of $\G$ must be assigned a vertex in some induced sub-graph of $\DG$
which is connected and of height $h$. 
For a fixed $\DG$, there are a fixed number of such subgraphs, 
and they all have one of three forms:
\begin{itemize}
\item An oriented path which is a subpath of $\Q_S$, for some set $S$, as defined in 
Definition~\ref{def:gammaextdual};
\item A collection of such oriented paths which all share their initial vertex (and no others);
\item A collection of such oriented paths which all share their final vertex (and no others).
\end{itemize}
In all three cases we can order the vertices of the subgraph by increasing level in $\DG$,
and within that by which path they belong to (when there is more than one), 
and within that by the distance along the path.
With the vertices ordered in this way, the subgraph  
has the property that for all edges $(a,b),(c,d)$ with $a < c$ we have $b \leq d$,
so it admits 
the binary polymorphisms $\min$ and $\max$.
Together with the fact that any unary cost function is submodular 
for any ordering of the domain, 
it follows that in all such cases the corresponding valued
constraint language is submodular, and an optimal solution can be found 
in polynomial time~\cite{cohen04:maximal}.

Finally, we consider the case when $\G$ is balanced and has the same height as $\DG$.
In this case only vertices at the top level in $\G$ can be assigned labels 
at the top level in $\DG$.
Let these vertices of $\G$ be $x_1,x_2,\ldots,x_q$. 
We will build an instance $\I'_d$ of $\VCSP(\GammaD)$ beginning 
with these vertices as variables.

If there is a unary constraint with cost function $\mu_\Gamma$ 
on any of these vertices in $\I'$,
then we add a unary constraint with cost function $\phi'_\Gamma$ in $\I'_d$,
where $\phi'_\Gamma$ is the unary cost function defined in Definition~\ref{def:gammadual}.
(Note that any other unary constraints on other variables in $\I'$ 
will not affect the cost of a feasible solution, 
because all other variables must be assigned a label with cost 0.) 

To complete the construction of $\I'_d$ we will add constraints of the form $match_{k,l}$ between pairs of vertices
$x_i$ and $x_j$ where it can be shown from 
the structure of $\G$ that they must be assigned labels that agree 
in positions $k$ and $l$ respectively.

To examine the structure of $\G$, consider the connected components of the induced subgraph
of $\G$ obtained by removing all vertices at the top level and all vertices at level 0.
Each such component is a balanced digraph of height at most $m$ 
which must be assigned labels from a single oriented path in $\DG$ of the form $\Q_S$,
for some set $S\subseteq\{1,...,m\}$. 
Note that the choice of oriented path in $\DG$ is fixed
by the assignment to any vertex in the component.  

For any such component $C$ there will be a unique smallest set $S_0\subseteq\{1,...,m\}$ 
such that any feasible solution to $\I'$ can assign labels to the variables in $C$ from
the oriented path $\Q_{S_0}$~\cite{Bulin15:lmcs}. 
Moreover, it is shown in~\cite{Bulin15:lmcs} that this set $S_0$ can be computed 
in polynomial time (in fact, in logarithmic space). 
For each component $C$ this set will be 
denoted\footnote{The notation used in~\cite{Bulin15:lmcs} is $\Gamma(C)$, but we
use a different notation here to avoid confusion with the valued constraint 
language $\Gamma$.}
by $S_0(C)$.

If there are edges in $\G$ from one such component $C$ 
to two distinct vertices $x_i$ and $x_j$ 
at the top level of $\G$, then these vertices must be assigned the same label
in any feasible solution to $\I'$, due to the structure of the paths in $\DG$,
so we add a constraint $match_{kk}(x_i,x_j)$ to $\I'_d$ for $k = 1,2,\ldots,m$.

Next, if there is an edge in $\G$ from some component $C$ to a vertex $x_i$ at the top level
of $\G$, and $S_0(C)$ contains two distinct indices $k$ and $l$, then the label 
assigned to $x_i$ in any feasible solution to $\I'$ must agree in positions $k$ and $l$.
Hence for each such case we add a constraint $match_{kl}(x_i,x_i)$ to $\I'_d$.

Next, if there is an edge in $\G$ from some vertex $y_1$ at level 0 to some component $C$,
and another edge in $\G$ from some vertex $y_2$ at level 0 to the same component $C$,
then we know that any feasible solution to $\I'$ must assign the same
label to $y_1$ and $y_2$, so we say that $y_1$ and $y_2$ are linked. 
Taking the reflexive, transitive closure of this linking relation gives an equivalence
relation on the vertices in $\G$ at level 0.

Finally, if there is an edge in $\G$ from a vertex $y_1$ at level 0 in $\G$ to a
component $C_1$, and an edge from $C_1$ to a vertex $x_i$ at the top level, and
there is also a vertex $y_2$ at level 0 which is equivalent to $y_1$, and an
edge from $y_2$ to a component $C_2$, and an edge from $C_2$ to a vertex $x_j$
at the top level in $\G$, then we proceed as follows: choose an index $k \in
S_0(C_1)$ and an index $l \in S_0(C_2)$ and add the constraint
$match_{kl}(x_i,x_j)$ to $\I'_d$. This ensures that the label assigned to $x_i$
in any feasible solution to $\I'$ must agree in position $k$ with the label
assigned to $x_j$ in position $l$.

Now we have constructed an instance $\I'_d$ in $\VCSP(\GammaD)$ whose constraints impose
precisely the same restrictions on feasible solutions as the binary constraints in $\I'$
(whose scopes are specified by the edges of $\G$). We have also imposed unary constraints
on the variables of $\I'_d$ to ensure that the cost of any feasible solution is the same as 
the cost of the corresponding feasible solution to $\I'$. 
Hence for any feasible solution to $\I'$ there will be a feasible solution to $\I'_d$ with 
the same cost, and vice versa, which gives the result.
\end{proof}

\subsection{Preservation of Algebraic Properties}
\label{balancedSec}

We now investigate how the polymorphisms of a valued constraint language $\Gamma$ 
(with finitely many cost functions) 
are related to the polymorphisms of the extended dual language $\GammaE$. In the
proof of Theorem~\ref{thm:extdualpol} we will closely follow results
from~\cite{Bulin15:lmcs}.
\begin{theorem}
\label{thm:extdualpol}
Let $\Gamma$ be any valued constraint language over $D$, such that $|\Gamma|$ is finite,
and let $\GammaE$ be the extended dual of $\Gamma$.
If $\Gamma$ is a rigid core, then  
$\{\Crestrict{\fE}{D} : \fE \in \pol(\GammaE)\} = \pol(\Gamma)$.

Moreover, for each $f \in \polk{\Gamma}$ there is at least one operation 
$\fE \in \polk{\GammaE}$ such that 
$\fE$ satisfies all linear balanced identities satisfied by $f$ and 
\begin{equation}
\label{eq:fedef}
\fE(\vec{x}_1,\ldots,\vec{x}_k)
\ =\ \begin{cases}
\extend{f}(\vec{x}_1,\ldots,\vec{x}_k) & \mbox{if each $\vec{x}_i \in D'$} \\
\mbox{some label not in $D'$}          & \mbox{otherwise}.
\end{cases}
\end{equation}
where $D'$ denotes the set $\Feas(\phi_\Gamma)\subseteq D^m$ for the single cost function
$\phi_\Gamma$ defined in Proposition~\ref{prop:single}.
\end{theorem}
\begin{proof}
First, consider any $\fE \in \pol(\GammaE)$. 
If we apply the extended dual construction given in 
Definition~\ref{def:extdual}, we obtain an instance $\IE$ of $\CSP(\GammaE)$
where in any feasible solution the variables at level 0 must take
values from $D$ that together form tuples from $\Feas(\phi_\Gamma)$. 
Hence $\fE|_D$ must be a polymorphism of $\Gamma$.  

For the converse, consider any $f \in \pol(\Gamma)$.
As noted in Section~\ref{sec:identities}, assuming that $\Gamma$ is a rigid core ensures
that every polymorphism of $\Gamma$ is idempotent. 
It is shown in~\cite[Proof of Theorem 5.1]{Bulin15:lmcs} 
that any idempotent polymorphism $f$ of the relation $\Feas(\phi_\Gamma)$ 
can be extended to a polymorphism $\fE$ of the associated digraph $\DG$ 
described in Definition~\ref{def:gammaextdual} 
that satisfies Equation~\ref{eq:fedef}.
Since {\em any} operation defined on the vertices of $\DG$ is a polymorphism of 
the unary finite-valued cost function $\mu_\Gamma$ 
described in Definition~\ref{def:gammaextdual},
the operation $\fE$ is a polymorphism of $\GammaE$.

Moreover, it is also shown in~\cite[Proof of Theorem 5.1]{Bulin15:lmcs} 
that $\fE$ satisfies many of the same identities as $f$, 
including all linear balanced identities 
that are satisfied by the polymorphisms of the zigzag.
By Lemma~5.3 of~\cite{Bulin15:lmcs}, \emph{all} balanced identities  
are satisfied by the polymorphisms of the zigzag,
so $\fE$ satisfies all linear balanced identities satisfied by $f$.
\end{proof}
For the special case of unary polymorphisms, we can say more:
Lemma~4.1 of \cite{Bulin15:lmcs} states that the unary polymorphisms of a relation
and of the corresponding digraph $\DG$ are in one-to-one correspondence. 
Hence, we immediately get the following.
\begin{lemma} 
\label{lem:extdualrigidcore}
Let $\Gamma$ be a valued constraint language such that $|\Gamma|$ is finite,
and let $\GammaE$ be the extended dual of $\Gamma$.
$\Gamma$ is a rigid core if and only $\GammaE$ is a rigid core.
\end{lemma}

Following our discussion in Section~\ref{sec:identities},
Theorem~\ref{thm:extdualpol} and Lemma~\ref{lem:extdualrigidcore} 
show that if a rigid core crisp language $\Gamma$ 
has the property of being solvable using local consistency methods 
then so does the associated binary language $\GammaE$.

Our next result shows that for finite rigid core valued constraint languages $\Gamma$,
the fractional polymorphisms of $\GammaE$ are closely related
to the fractional polymorphisms of $\Gamma$.
\begin{theorem}
\label{thm:samefpol2}
Let $\Gamma$ be a valued constraint language such that $|\Gamma|$ is finite 
and let $\GammaE$ be the extended dual of $\Gamma$.

If $\Gamma$ is a rigid core, then for any fractional polymorphism $\omega$ of $\Gamma$
there is a corresponding 
fractional polymorphism $\omega_e$ of $\GammaE$
such that 
for each $f \in \supp(\omega)$ there is a corresponding $\fE \in \supp(\omega_e)$
and vice versa. 
Moreover, $\omega_e(\fE) = \omega(f)$ for all $\fE \in \supp(\omega_e)$.
\end{theorem}
\begin{proof}
By Proposition~\ref{prop:single} we may assume that $\Gamma$ consists of a single
cost function $\phi_\Gamma: D^m \rightarrow \qq$, where $\Feas(\phi_\Gamma) = D'$.

Now consider any function $\omega:\polk{\Gamma} \rightarrow \qnn \in \fpolk{\Gamma}$.
By Theorem~\ref{thm:extdualpol}, for each $f \in \polk{\Gamma}$ we can choose a 
corresponding $f_e \in \polk{\GammaE}$ satisfying Equation~\ref{eq:fedef}.
Hence we can define a function $\omega_e:\polk{\GammaE} \rightarrow \qnn$ 
by setting $\omega_e(\fE) = \omega(f)$ for all $f \in \polk{\Gamma}$ 
(and setting all other values of $\omega_e$ to zero).

To check that $\omega_e$ is a fractional polymorphism of $\GammaE$ we only need
to verify that it satisfies Equation~\ref{eq:fpol} in Definition~\ref{def:fpol} 
for each cost function in $\GammaE$.

The language $\GammaE$ contains just 
the binary relation $\DG$ and the unary cost function $\mu_\Gamma$,
as specified in Definition~\ref{def:gammaextdual}.
As $\DG$ is a relation, and each $\fE$ is a polymorphism of $\DG$,
the inequality in Definition~\ref{def:fpol}
is trivially satisfied by $\DG$ (both sides are equal to zero).

It remains to show that $\omega_e$ is a fractional polymorphism of $\mu_\Gamma$.
When applied to $\mu_\Gamma$, this condition says that, for any
$x_1,\ldots,x_k$ in the domain of $\GammaE$, we must have 
\[
\sum_{\fE\in \polk{\GammaE}} \omega_e(\fE)\mu_\Gamma(\fE(x_1,...,x_k)) 
\leq \frac{1}{k} (\mu_\Gamma(x_1) + \ldots + \mu_\Gamma(x_k)).
\]
Recall that, by definition, $\mu_\Gamma(x)=0$ for all $x \not \in D'$. 
By Theorem~\ref{thm:extdualpol}, it follows that if
$\fE(x_1,\ldots,x_k)\in D'$ then $x_1,\ldots,x_k\in D'$. 
Thus, if not all $x_1,\ldots,x_k$ are in $D'$, 
the only possible non-zero terms appear in the RHS of the inequality, 
and hence it is trivially true. 

On the other hand, if all
$x_1,\ldots,x_k$ are in $D'$ then we have $\mu_\Gamma(x_i) = \phi_\Gamma(x_i)$
for $i = 1,2,\ldots,k$. 
Moreover, since $\fE(x_1,\ldots,x_k) = f(x_1,\ldots,x_k)$, 
and $f$ is a polymorphism of $\Gamma$, we have 
$\fE(x_1,\ldots,x_k)\in D'$ and so
$\mu_\Gamma(\fE(x_1,\ldots,x_k))=\phi_\Gamma(f(x_1,\ldots,x_k))$. 
In this case, the inequality holds because the inequality 
\[
\sum_{f\in \polk{\Gamma}} \omega(f)\phi_\Gamma(f(x_1,...,x_k)) \leq
\frac{1}{k} (\phi_\Gamma(x_1) + \ldots + \phi_\Gamma(x_k))
\]
holds for $\omega$,
because it is a fractional polymorphism of $\Gamma$.
\end{proof}

Following our discussion in Section~\ref{sec:identities}, combining
Theorem~\ref{thm:extdualpol} with Theorem~\ref{thm:samefpol2} shows that the
property of being solvable using the basic linear programming relaxation 
is possessed by the binary language $\GammaE$ if it is possessed by $\Gamma$.
Similarly, the property of being solvable by 
constant levels of the Sherali-Adams linear programming relaxations 
is possessed by the binary language $\GammaE$ if it is possessed by $\Gamma$.

\subsection{Reduction to Minimum Cost Homomorphism}

We have shown that for any valued constraint language with a finite number of cost functions
of arbitrary arity 
we can construct an equivalent language with a single unary cost function and a single binary
crisp cost function. 

Valued constraint problems with a single binary crisp cost function, 
described by a digraph $\H$, can also be seen as graph homomorphism problems.
In a graph homomorphism problem we are given an instance specified by a digraph $\G$
and asked whether there is a mapping from the vertices of $\G$ to the vertices
of a fixed digraph $\H$ such that adjacent vertices in $\G$ are mapped to adjacent
vertices in $\H$. 
Such a mapping is called a homomorphism from $\G$ to $\H$.

If we have a VCSP instance $\I$ 
over a language containing only a single binary relation $\H$,
then it is easy to check that the feasible solutions to $\I$ are precisely the 
homomorphisms from $\G$ to $\H$, where $\G$ is the digraph whose edges are the scopes
of the constraints in $\I$.

If our instance $\I$ also has unary finite-valued cost functions, then it is 
equivalent to the so-called Minimum Cost Homomorphism Problem~\cite{Gutin:mincostdichotomy},
where the cost of a homomorphism is defined by a unary function on each vertex of the input
that assigns a cost to each possible vertex of the target digraph.
The Minimum Cost Homomorphism Problem for a fixed digraph $\H$ is denoted $\MCHom(\H)$.
The special case where all the unary finite-valued cost functions are
chosen from some fixed set $\Delta$ is denoted $\MCHom(\H,\Delta)$.

The problem $\MCHom(\H)$ was studied in a series of papers, and complete
complexity classifications were given in~\cite{Gutin:mincostdichotomy} for
undirected graphs, in~\cite{Hell12:sidma} for digraphs, and
in~\cite{Tak10:MinCostHom} for more general structures. Partial complexity
classifications for the problem $\MCHom(\H,\Delta)$ were obtained
in~\cite{Tak12:extMinCostHom,Uppman13:icalp,Uppman14:mincosthom}. One can see
that MinCostHom is an intermediate problem between CSP and VCSP, as there is an
optimisation aspect, but it is limited in the sense that it is controlled by
separate unary cost functions, without explicit interactions of variables.

By Theorem~\ref{thm:extdualequiv} and Lemma~\ref{lem:extdualrigidcore} we obtain the
following corollary, which shows that a very restricted case of binary 
MinCostHom can express all valued constraint problems.
\begin{corollary}
Let $\Gamma$ be a valued constraint language such that $|\Gamma|$ is finite 
and $\Gamma$ is a rigid core. There is a balanced
digraph $\DG$ which is a rigid core and a finite-valued unary cost function $\mu_\Gamma$ 
such that problems $\VCSP(\Gamma)$ and $\MCHom(\DG,\{\mu_\Gamma\})$ 
are polynomial-time equivalent.
\end{corollary}

An interesting problem is to characterise which digraph homomorphism problems 
can capture $\NP$-hard VCSPs. 
For the restricted case of ordinary CSPs the following result is known.

\begin{theorem}[\hspace{-0.001ex}\cite{Feder98:monotone}]
Every CSP is polynomial-time equivalent to a balanced digraph homomorphism
problem with only 5 levels. 
\end{theorem}
Recall that an $n$-level digraph has height $n-1$.
We remark that \cite{Feder98:monotone} also shows that the digraph homomorphism problem for a
balanced digraph with 4 levels is solvable in polynomial time.

To illustrate how the digraph homomorphism problem can capture NP-hard VCSPs 
we give an example of a 5-level digraph and unary weighted relation which can capture
Max-Cut, a canonical $\NP$-hard VCSP.
\begin{example}\label{hardDigraph}
Consider the digraph $\H$ shown Figure~\ref{digraphEx2}.
Let the unary weighted relation $\mu(v)$ be \[\mu(v)=\left\{\begin{array}{l l}1 &
\quad\text{if }v=b\text{ or }v=c\\0 & \quad \text{otherwise}\end{array}\right.\]
for every vertex $v\in V^{\H}$.

Now consider the instance of $\MCHom(\H,\{\mu\})$ with the source digraph $\G$ shown
in Figure~\ref{digraphSourceEx2} and the unary cost function $\mu$ applied to 
all vertices of $\G$.
It is straightforward to check that the homomorphism that
maps $x\rightarrow 0$ and $y\rightarrow 1$ has cost 0, as does the homomorphism
that maps $x\rightarrow 1$ and $y\rightarrow 0$. However the homomorphism that
maps $x\rightarrow 0$ and $y\rightarrow 0$ has cost 2, and likewise for the
homomorphism that maps $x\rightarrow 1$ and $y\rightarrow 1$. If we consider
these homomorphisms as the possible assignments of labels to the variables we
have a VCSP instance $\I$ with $\Phi_\I(0,0) = \Phi_I(1,1) > \Phi_I(0,1) =
\Phi_I(1,0)$, and thus we capture Max-Cut.

Note that following the construction in Definition~\ref{def:gammaextdual}
for any binary finite-valued cost function $\phi_\Gamma$ we obtain a digraph $\DG$
which is quite similar to $\H$, except for an additional oriented path from $0$ to $c$ and another 
from $1$ to $b$, each consisting of a single edge followed by two zigzags
and another single edge. However no path in $\G$ can possibly map onto these
oriented paths, so they are omitted from $\H$ to simplify the diagram. 
\end{example}

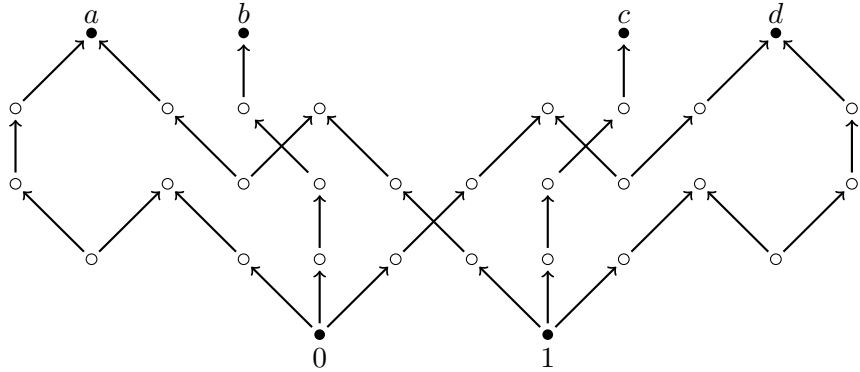
\begin{figure}
\begin{center}
\begin{tikzpicture}

\fill (canvas cs:x=2cm,y=5cm) circle (2pt);
\fill (canvas cs:x=11cm,y=5cm) circle (2pt);
\fill (canvas cs:x=5cm,y=1cm) circle (2pt);
\fill (canvas cs:x=8cm,y=1cm) circle (2pt);

\draw (canvas cs:x=1cm,y=4cm) circle (2pt);
\draw (canvas cs:x=3cm,y=4cm) circle (2pt);
\draw (canvas cs:x=5cm,y=4cm) circle (2pt);
\draw (canvas cs:x=8cm,y=4cm) circle (2pt);
\draw (canvas cs:x=10cm,y=4cm) circle (2pt);
\draw (canvas cs:x=12cm,y=4cm) circle (2pt);

\draw (canvas cs:x=1cm,y=3cm) circle (2pt);
\draw (canvas cs:x=3cm,y=3cm) circle (2pt);
\draw (canvas cs:x=4cm,y=3cm) circle (2pt);
\draw (canvas cs:x=6cm,y=3cm) circle (2pt);
\draw (canvas cs:x=7cm,y=3cm) circle (2pt);
\draw (canvas cs:x=9cm,y=3cm) circle (2pt);
\draw (canvas cs:x=10cm,y=3cm) circle (2pt);
\draw (canvas cs:x=12cm,y=3cm) circle (2pt);

\draw (canvas cs:x=2cm,y=2cm) circle (2pt);
\draw (canvas cs:x=4cm,y=2cm) circle (2pt);
\draw (canvas cs:x=6cm,y=2cm) circle (2pt);
\draw (canvas cs:x=7cm,y=2cm) circle (2pt);
\draw (canvas cs:x=9cm,y=2cm) circle (2pt);
\draw (canvas cs:x=11cm,y=2cm) circle (2pt);

\draw (canvas cs:x=5cm,y=2cm) circle (2pt);
\draw (canvas cs:x=5cm,y=3cm) circle (2pt);
\draw (canvas cs:x=4cm,y=4cm) circle (2pt);
\fill (canvas cs:x=4cm,y=5cm) circle (2pt);
\draw (canvas cs:x=8cm,y=2cm) circle (2pt);
\draw (canvas cs:x=8cm,y=3cm) circle (2pt);
\draw (canvas cs:x=9cm,y=4cm) circle (2pt);
\fill (canvas cs:x=9cm,y=5cm) circle (2pt);

\draw [->,thick](5,1.15) -- (5,1.85);
\draw [->,thick](5,2.15) -- (5,2.85);
\draw [->,thick](4.85,3.15) -- (4.15,3.85);
\draw [->,thick](4,4.15) -- (4,4.85);
\draw [->,thick](8,1.15) -- (8,1.85);
\draw [->,thick](8,2.15) -- (8,2.85);
\draw [->,thick](8.15,3.15) -- (8.85,3.85);
\draw [->,thick](9,4.15) -- (9,4.85);
\node [above] at (4,5) {$b$};
\node [above] at (9,5) {$c$};

Level 0 vertex labels
\node [below] at (5,0.95) {$0$};
\node [below] at (8,0.95) {$1$};
Top level vertex labels
\node [above] at (2,5) {$a$};
\node [above] at (11,5) {$d$};

Path 0 -> (1,0)
\draw [->,thick](1.1,4.1) -- (1.9,4.9);
\draw [->,thick](1,3.15) -- (1,3.85);
\draw [->,thick](1.9,2.1) -- (1.1,2.9);
\draw [->,thick](2.1,2.1) -- (2.9,2.9);
\draw [->,thick](3.9,2.1) -- (3.1,2.9);
\draw [->,thick](4.9,1.1) -- (4.1,1.9);

Path 0 -> (0,1)
\draw [->,thick](10.1,4.1) -- (10.9,4.9);
\draw [->,thick](9.1,3.1) -- (9.9,3.9);
\draw [->,thick](8.9,3.1) -- (8.1,3.9);
\draw [->,thick](7.1,3.1) -- (7.9,3.9);
\draw [->,thick](6.1,2.1) -- (6.9,2.9);
\draw [->,thick](5.1,1.1) -- (5.9,1.9);

Path 1 -> (1,0)
\draw [->,thick](2.9,4.1) -- (2.1,4.9);
\draw [->,thick](3.9,3.1) -- (3.1,3.9);
\draw [->,thick](4.1,3.1) -- (4.9,3.9);
\draw [->,thick](5.9,3.1) -- (5.1,3.9);
\draw [->,thick](6.9,2.1) -- (6.1,2.9);
\draw [->,thick](7.9,1.1) -- (7.1,1.9);

Path 1 -> (0,1)
\draw [->,thick](11.9,4.1) -- (11.1,4.9);
\draw [->,thick](12,3.15) -- (12,3.85);
\draw [->,thick](11.1,2.1) -- (11.9,2.9);
\draw [->,thick](10.9,2.1) -- (10.1,2.9);
\draw [->,thick](9.1,2.1) -- (9.9,2.9);
\draw [->,thick](8.1,1.1) -- (8.9,1.9);

\end{tikzpicture}

\caption{The target digraph $\H$ of Example~\ref{hardDigraph}.}
\label{digraphEx2}
\end{center}
\end{figure}

\begin{figure}
\begin{center}
\begin{tikzpicture}

\fill (canvas cs:x=2cm,y=5cm) circle (2pt);
\fill (canvas cs:x=11cm,y=5cm) circle (2pt);
\fill (canvas cs:x=5cm,y=1cm) circle (2pt);
\fill (canvas cs:x=8cm,y=1cm) circle (2pt);

\draw (canvas cs:x=1cm,y=4cm) circle (2pt);
\draw (canvas cs:x=3cm,y=4cm) circle (2pt);
\draw (canvas cs:x=5cm,y=4cm) circle (2pt);
\draw (canvas cs:x=8cm,y=4cm) circle (2pt);
\draw (canvas cs:x=10cm,y=4cm) circle (2pt);
\draw (canvas cs:x=12cm,y=4cm) circle (2pt);

\draw (canvas cs:x=1cm,y=3cm) circle (2pt);
\draw (canvas cs:x=3cm,y=3cm) circle (2pt);
\draw (canvas cs:x=4cm,y=3cm) circle (2pt);
\draw (canvas cs:x=6cm,y=3cm) circle (2pt);
\draw (canvas cs:x=7cm,y=3cm) circle (2pt);
\draw (canvas cs:x=9cm,y=3cm) circle (2pt);
\draw (canvas cs:x=10cm,y=3cm) circle (2pt);
\draw (canvas cs:x=12cm,y=3cm) circle (2pt);

\draw (canvas cs:x=2cm,y=2cm) circle (2pt);
\draw (canvas cs:x=4cm,y=2cm) circle (2pt);
\draw (canvas cs:x=6cm,y=2cm) circle (2pt);
\draw (canvas cs:x=7cm,y=2cm) circle (2pt);
\draw (canvas cs:x=9cm,y=2cm) circle (2pt);
\draw (canvas cs:x=11cm,y=2cm) circle (2pt);

Level 0 vertex labels
\node [below] at (5,0.95) {$x$};
\node [below] at (8,0.95) {$y$};

Path 0 -> (1,0)
\draw [->,thick](1.1,4.1) -- (1.9,4.9);
\draw [->,thick](1,3.15) -- (1,3.85);
\draw [->,thick](1.9,2.1) -- (1.1,2.9);
\draw [->,thick](2.1,2.1) -- (2.9,2.9);
\draw [->,thick](3.9,2.1) -- (3.1,2.9);
\draw [->,thick](4.9,1.1) -- (4.1,1.9);

Path 0 -> (0,1)
\draw [->,thick](10.1,4.1) -- (10.9,4.9);
\draw [->,thick](9.1,3.1) -- (9.9,3.9);
\draw [->,thick](8.9,3.1) -- (8.1,3.9);
\draw [->,thick](7.1,3.1) -- (7.9,3.9);
\draw [->,thick](6.1,2.1) -- (6.9,2.9);
\draw [->,thick](5.1,1.1) -- (5.9,1.9);

Path 1 -> (1,0)
\draw [->,thick](2.9,4.1) -- (2.1,4.9);
\draw [->,thick](3.9,3.1) -- (3.1,3.9);
\draw [->,thick](4.1,3.1) -- (4.9,3.9);
\draw [->,thick](5.9,3.1) -- (5.1,3.9);
\draw [->,thick](6.9,2.1) -- (6.1,2.9);
\draw [->,thick](7.9,1.1) -- (7.1,1.9);

Path 1 -> (0,1)
\draw [->,thick](11.9,4.1) -- (11.1,4.9);
\draw [->,thick](12,3.15) -- (12,3.85);
\draw [->,thick](11.1,2.1) -- (11.9,2.9);
\draw [->,thick](10.9,2.1) -- (10.1,2.9);
\draw [->,thick](9.1,2.1) -- (9.9,2.9);
\draw [->,thick](8.1,1.1) -- (8.9,1.9);

\end{tikzpicture}

\caption{The source digraph $\G$ of Example~\ref{hardDigraph}.}
\label{digraphSourceEx2}
\end{center}
\end{figure}
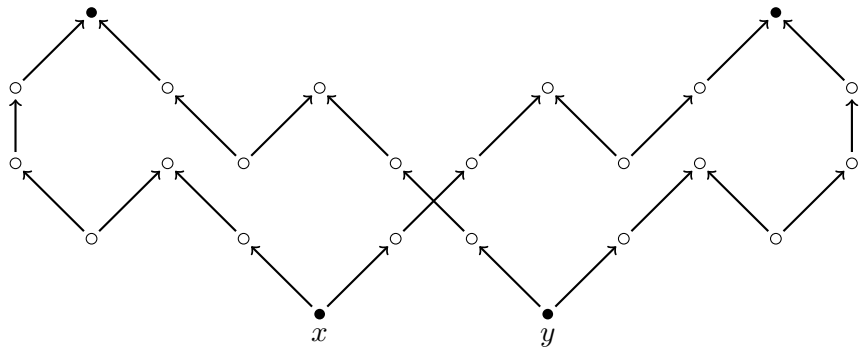

\section{Conclusion}

Transforming a constraint satisfaction problem to a binary problem has a number
of advantages and disadvantages which have been investigated by many
authors~\cite{Rossi90:equivalence,Feder98:monotone,Bacchus2002:aij,Stergiou05:jair,Atserias08:digraph,Bulin15:lmcs}.
Such a transformation changes many aspects of the problem, such as what
inferences can be derived by various kinds of propagation. One might expect that
achieving the simplicity of a binary representation would incur a corresponding
increase in the sophistication of the required solving algorithms.

However, we have shown here that the well-known dual encoding of the
VCSP converts any finite language, $\Gamma$, of arbitrary arity to a
\emph{binary} language, $\GammaD$, of a very restricted kind, such
that there is a bijection between the polymorphisms of $\Gamma$ and
the polymorphisms of $\GammaD$, and the corresponding polymorphisms satisfy exactly the
same identities. Hence we have shown that the
algebraic analysis of valued constraint languages can focus on a very
restricted class of binary languages (at least in the case of finite languages).
Moreover, many important
algorithmic properties, such as the ability to solve problems using a
bounded level of consistency, or by a linear programming relaxation, are also preserved
by the dual encoding.

Furthermore, we have adapted the recently obtained reduction for
CSPs~\cite{Bulin15:lmcs} to VCSPs and thus obtained a polynomial-time
equivalence between VCSPs and MinCostHom problems. 
In order to study families of valued constraint languages with finitely many cost functions
defined by
fractional polymorphisms satisfying linear balanced identities, we now know
that we need only study MinCostHom problems. This is important since, for
example, to prove the algebraic dichotomy conjecture for core crisp languages
we only need to study polymorphisms satisfying linear balanced
identities~\cite{Barto:wonderland}.

We remark that the CSP reduction from~\cite{Bulin15:lmcs} is shown to preserve a
slightly larger class of identities than that of linear balanced identities, and
works not only in polynomial time but actually in \emph{logarithmic space}. We
believe that our extension of this reduction can also be
adapted to derive similar conclusions, but we leave this as an open problem.
Our contribution is to show that the CSP reduction from~\cite{Bulin15:lmcs} 
can be extended to the more general setting of the VCSP, 
and that the extended reduction preserves all linear balanced identities.
Finally we remark that, even in the more general setting of the VCSP,
using the dual construction as a stepping stone considerably simplifies the proof.


\newcommand{\noopsort}[1]{}\newcommand{\Zivny}{\noopsort{ZZ}\v{Z}ivn\'y}

\end{document}